\documentclass[11pt]{article}
\usepackage{amsmath}
\usepackage{amssymb}
\usepackage{amsfonts}
\usepackage{url}
\usepackage{amsthm}
\usepackage{color}
\usepackage{fullpage}

\usepackage[dvipdfm,%
bookmarks=true,%
bookmarksnumbered=true,%
bookmarksopen=true,%
colorlinks=true,%
linkcolor=blue,%
citecolor=blue,%
filecolor=blue,%
urlcolor=blue]{hyperref}

\newtheorem{theorem}{Theorem}
\newtheorem{corollary}{Corollary}
\newtheorem{lemma}{Lemma}

\newtheorem{proposition}{Proposition}
\newtheorem{definition}{Definition}

\newcommand{\ignore}[1]{}
\newcommand{\magicwand}{\itemsep=1pt\parskip=1pt\topskip=0pt}

\newcommand{\F}{\mathbb{F}}
\newcommand{\N}{\mathbb{N}}

\newcommand{\bin}{\{0, 1\}}

\newcommand{\Supp}{\mathrm{Supp}}

\newcommand{\Enc}{\textsf{Enc}} 
\newcommand{\Dec}{\textsf{Dec}} 
\newcommand{\Com}{\textsf{Com}} 
\newcommand{\Decom}{\textsf{Decom}} 

\newcommand{\Rec}{\textsf{Rec}}

\newcommand{\calO}{\mathcal{O}}
\newcommand{\calF}{\mathcal{F}}
\newcommand{\calH}{\mathcal{H}}
\newcommand{\invert}{\mathsf{invert}}
\newcommand{\forge}{\mathsf{forge}}
\newcommand{\correctf}{\mathsf{correctf}}
\newcommand{\invertible}{\mathsf{invertible}}
\newcommand{\forgeable}{\mathsf{forgeable}}
\newcommand{\Exp}{\mathbb{E}}
\newcommand{\SD}{\mathsf{SD}}

\begin{document}


\title{Error Correction by Structural Simplicity: \\Correcting  Samplable Additive Errors}

\author{Kenji Yasunaga\\
Kanazawa University\\
yasunaga@se.kanazawa-u.ac.jp}
\date{\today}





\maketitle

\begin{abstract}
  This paper explores the possibilities and limitations of error correction by the structural simplicity of error mechanisms.
  Specifically, we consider channel models, called \emph{samplable additive channels}, in which
   (a) errors are efficiently sampled without the knowledge of the coding scheme or the transmitted codeword;
  (b) the entropy of the error distribution is bounded; and
   (c) the number of errors introduced by the channel is unbounded.
  For the channels, several negative and positive results are provided.
   Assuming the existence of one-way functions, there are samplable additive errors of entropy $n^{\epsilon}$ for $\epsilon \in (0,1)$
  that are pseudorandom, and thus not correctable by efficient coding schemes.
  It is shown that there is an oracle algorithm that induces a samplable distribution over $\bin^n$ of entropy $m = \omega( \log n)$
  that is not pseudorandom, but is uncorrectable by efficient schemes of rate less than $1 - m/n - o(1)$.
  The results indicate that restricting error mechanisms to be efficiently samplable and not pseudorandom is insufficient for error correction.
    As positive results, some conditions are provided under which efficient error correction is possible.
\end{abstract}


\section{Introduction}

In the theory of error-correcting codes,
two of the most-studied channel models are probabilistic channels and worst-case channels.
In probabilistic channels, errors are introduced through stochastic processes,
and the most well-known one is the binary symmetric channel (BSC).
In worst-case (or adversarial) channels, errors are introduced adversarially 
by considering the choice of \textcolor{black}{coding schemes}
and transmitted \textcolor{black}{strings, called \emph{codewords}},
under the restriction of the \emph{error rate}, \textcolor{black}{which is
  the ratio of the number of errors to the length of the codeword.}
In his seminal work~\cite{Sha48}, Shannon showed that reliable communication can be achieved over BSC
if the \emph{coding rate}\textcolor{black}{, which represents the efficiency of transmission,}
is less than $1-h_2(p)$, where
  $h_2(\cdot)$ is the binary entropy function,
and $p$ is the crossover probability of BSC.
In contrast, it is known that reliable communication cannot be achieved over worst-case channels
when the error rate is at least $1/4$ unless the coding rate tends to zero~\cite{Plo60}.

If we view the introduction of errors as \emph{computation} of the channel,
probabilistic channels perform low-cost computation with little knowledge about the \textcolor{black}{coding scheme} and the \textcolor{black}{transmitted codeword}, 
while worst-case channels perform high-cost computation with the full-knowledge.
As intermediate channels between probabilistic channels and worst-case channels,
Lipton~\cite{Lip94} introduced \emph{computationally-bounded channels},
where errors are introduced by polynomial-time computation. 
\textcolor{black}{It has been demonstrated that reliable communication for these intermediate channels
  can be achieved by more efficient schemes than those for worst-case ones~\cite{Lip94,MPSW10,GS16}.}
\ignore{
He showed that reliable communication can be achieved
at the coding rate less than $1-h_2(p)$ in the shared randomness setting, where $p < 1$ is the error rate,
which is the fraction of errors introduced by the channel.
Micali et al.~\cite{MPSW10} presented reliable coding schemes in the public-key infrastructure setting. 
Guruswami and Smith~\cite{GS16} showed reliable coding schemes without assuming the shared randomness or
the public-key infrastructure.
Note that these work~\cite{Lip94,MPSW10,GS16} consider the settings in which channels are computationally-bounded
and the error rate is bounded.
}

Although there have been many studies on error correction in various channel models,
the most basic principle for error correction has been the same:
the number of errors occurred in communications is restricted.
Namely, an upper bound on the error rate is a priori provided, and coding schemes are designed with the knowledge of the bound.
However, since qualitatively better error correction is possible for computationally-bounded channels,
it may be possible to correct errors without resorting to the bound on the error rate.
Thus, we ask the following question:
\begin{quote}
\emph{Is it possible to correct errors based on the structural simplicity of error mechanisms?}
\end{quote}

In this work, we partially answer the above question.
Specifically, we introduce a novel channel model, called \emph{samplable additive channel},
and investigate the error-correction capabilities for the channel.
In samplable additive channels,
  (a) errors are efficiently sampled without the knowledge of the coding scheme or the transmitted codeword;
  (b) the entropy of the error distribution is bounded; and
  (c) the error rate is unbounded.
Condition (a) means that there is a polynomial-time algorithm that generates samples according to the error distribution,
where the algorithm is designed without the knowledge of the coding scheme, and does not receive the transmitted codeword as input.
The restriction of efficient samplability is also employed in the previous studies of computationally-bounded channels.
For a structural simplicity, condition (b) is necessary since
high-entropy distributions can generate unpredictable and complex errors.
Condition (c) is employed for removing the effect of bounding the error rate for error correction.
In addition, for the ease of analysis, we mainly consider \emph{flat} distributions, which are the uniform distributions over the supports.


Samplable additive channels are relatively simple channel models since the error distributions are identical
for every coding scheme and transmitted codeword.
The binary symmetric channel is an example.
The setting is incomparable to previous notions of error correction against computationally-bounded channels.
Our model is stronger because we do not restrict the error rate, 
but is weaker because the channel cannot see the code or the transmitted codeword.

\subsection{The Setting}

Before presenting our results,
we briefly describe our problem setting.

The coding scheme, also referred to as code, consists of two functions $\Enc : \bin^{Rn} \to \bin^n$ and $\Dec : \bin^n \to \bin^{Rn}$, where $R \in [0,1]$
is called \emph{coding rate}, or simply \emph{rate}.
A message $x \in \bin^{Rn}$ is encoded to $\Enc(x)$, 
and transmitted to the channel.
The channel introduces an error $z \in \bin^n$, and the decoder receives the string $\Enc(x) + z$,
where $+$ is the bit-wise addition modulo $2$, and outputs $\tilde{x} = \Dec(\Enc(x)+z)$.
Decoding is done successfully if $\tilde{x} = x$.
It is desirable to construct a high-rate code  that can successfully correct errors.
The \emph{error rate} is $w_z/n \in [0,1]$, where $w_z$ is the number of non-zero elements in $z$.

For every samplable additive channel, an error distribution $Z$ over $\bin^n$ is associated.
If $Z$ is flat, each element in the support of $Z$ is sampled with probability $1/|Z|$, and thus the entropy of $Z$ is $\log |Z|$,
which takes values in $[0,n]$. 
In our setting, we assume that sampling from $Z$ can be simulated by a polynomial-time algorithm,
the entropy of $Z$ is bounded, and the error rate is not bounded.
In addition, the coding scheme can be chosen based on the knowledge of $Z$,
and an error vector $z$ is sampled from $Z$ without knowing $\Enc(x)$.
We would like to know which $Z$ is (un)correctable, especially by efficient coding schemes.

\subsection{Main Results}

We investigate the error correction capabilities of samplable additive errors.

\paragraph{Errors from flat distributions.}
By simple combinatorial arguments, it can be shown that any flat $Z$ can be corrected by
a code of rate $R \leq 1 - m/n - o(1)$, but cannot for rate $R > 1 - m/n -o(1)$,
where $m$ is the entropy of $Z$. Thus, the rate $1-m/n$ is essentially optimal.
In addition, if $m = O(\log n)$ and $R \leq 1 -m/n - o(1)$, both the encoding and the decoding can be done in polynomial time.

\paragraph{Pseudorandom errors.}
Assuming the existence of one-way functions, there exists pseudorandom generators~\cite{HILL99},
which generate distributions that look random to every efficient algorithm.
Hence, no efficient scheme can correct errors from such distributions.
It follows from this fact that, assuming the existence of one-way functions,
there exists an error distribution $Z$ with entropy $n^\epsilon$ for any constant $\epsilon \in (0,1)$
  that are not efficiently correctable. 

  \paragraph{Errors with membership test.}
To avoid the impossibility of correcting pseudorandom errors,
we consider samplable distributions for which membership test can be done efficiently.
Such distributions are not pseudorandom since the membership test can be used to distinguish them from the uniform distribution.

We show the existence of an uncorrectable distribution with membership test for \emph{low-rate} codes.
As sampling algorithms, we employ \emph{oracle algorithms}~\cite{AB09},
  which can make a black-box use of an external entity called \emph{oracle}.
  We present an oracle algorithm that induces
a samplable distribution $Z$ of entropy $\textcolor{black}{m = }$ $\omega(\log n)$
that is not correctable by efficient coding schemes of rate $R < 1 - \textcolor{black}{m}/n - \textcolor{black}{o(1)}$.
The result complements the impossibility of correcting flat distributions for rate $R > 1 - \textcolor{black}{m}/n - o(1)$.
Also, the entropy of $\omega(\log n)$ is optimal since
any flat $Z$ with entropy $ O(\log n)$ can be corrected by a polynomial-time coding scheme.

\paragraph{Positive results.}
As a positive result, we show that if the set of error vectors forms a linear subspace,
then every error in the set can be corrected by an efficient coding scheme with optimal rate.

Also, we derive a computational condition under which
samplable additive errors can be corrected.
Intuitively, the condition is that a variant of the sampling algorithm of $Z$ is efficiently ``invertible''.
See Section~\ref{sec:compcond} for the details.
The result implies that the existence of one-way functions is necessary for proving the impossibility results for correcting samplable errors.

\begin{table*}
  \centering
\caption{Correctability of Samplable Additive Error $Z$}\label{tb:result}
\smallskip
\begin{tabular}{clll}
\hline\hline
\textbf{ Entropy $m$} & \textbf{Correctabilities} & \textbf{Assumptions} & \textbf{References} \\ \hline\hline

& $\forall$ flat $Z$, & & \\
--- & (1) $\exists$ code correcting $Z$ for   & None & Proposition~\ref{prop:random}\\
& \, \quad $R \leq 1 - \textcolor{black}{{m}/{n}} - \textcolor{black}{o(1)}$ in time $O(n^22^m) $ & & \\
& (2) not correctable for $R > 1 - \textcolor{black}{m/n} \textcolor{black}{- o(1)}$ & None & Proposition~\ref{prop:flatbound} \\ \hline


$n^\epsilon$ &   & & \\ 
\raisebox{0.1ex}[0pt]{$(0 < \epsilon < 1)$} & \raisebox{1.2ex}[0pt]{$\exists \, Z$  not efficiently correctable for any $R$}
& \raisebox{1.2ex}[0pt]{OWF} & \raisebox{1.2ex}[0pt]{Proposition~\ref{prop:prg}}
\\\hline



& $\exists$ $Z$ with membership test, not efficiently & & \\
\raisebox{1.2ex}[0pt]{$\omega(\log n)$} &  correctable for $R < 1 - \textcolor{black}{m/n} - o(1)$ & \raisebox{1.2ex}[0pt]{Oracle access} & \raisebox{1.2ex}[0pt]{Corollary~\ref{cor:lowentropy}} \\ \hline

& $\forall$ flat $Z$ over a linear subspace of dim. $m$, & & \\ 
\raisebox{1.2ex}[0pt]{---} & \raisebox{0.1ex}[0pt]{$\exists$ code correcting $Z$ for $R \leq 1-\textcolor{black}{{m}/{n}}$} & \raisebox{1.2ex}[0pt]{None} & \raisebox{1.2ex}[0pt]{Proposition~\ref{prop:subspace}}
\\\hline

--- & $\forall$ flat $Z = f(U_r)$  is efficiently correctable for  & $g$ is not &  \\
& $R \leq 1 - \textcolor{black}{m/n} - \textcolor{black}{o(1)}$ & distOWF & \raisebox{1.2ex}[0pt]{Theorem~\ref{thm:distow}}\\
\hline\hline



\end{tabular}
\end{table*}

We summarize our main results in Table~\ref{tb:result},
where $R$ denotes the rate of coding schemes, and $m$ the entropy of $Z$.

\ignore{
\paragraph{Necessity of one-way functions.}

Finally, we show that it is difficult to
prove \emph{unconditional} impossibility results 
for coding schemes of rate $R \leq 1 - \textcolor{black}{m}/n$, \textcolor{black}{where $m$ is the entropy of $Z$}.
Specifically, we show that if one-way functions do not exist, then any samplable flat distribution
is correctable by an efficient coding scheme of rate $1-\textcolor{black}{m}/n-o(1)$.
Thus, it is necessary to assume the existence of one-way functions or oracle access to derive 
impossibility results for rate $R  \leq 1 - \textcolor{black}{m}/n$.

The results are summarized in Table~\ref{tb:result},
where $R$ denotes the rate of coding schemes.
}

\subsection{Related Work}

The notion of computationally-bounded channel was introduced by Lipton~\cite{Lip94}.
He showed that if the sender and the receiver can share secret randomness,
then the Shannon capacity can be achieved for this channel.
Micali et al.\,\cite{MPSW10} considered a similar channel model in a public-key setting, 
and gave a coding scheme based on list-decodable codes and digital signature.
Guruswami and Smith~\cite{GS16} gave constructions of
capacity achieving codes for worst-case additive-error channel 
and time/space-bounded channels.
They also gave an impossibility result for bit-fixing channels,
but their result can be applied to channels that use the information on the code and the transmitted codeword.
Shaltiel and Silbak~\cite{SS16} gave explicit list-decodable codes for computationally-bounded channels
based on complexity assumptions.
Cryptographic treatment of codes against computationally-bounded channels was studied in~\cite{Yas16}.
Note that all the previous work of computationally-bounded channels assumes that
the error rate is bounded above by some constant $p \in [0,1]$,
and codes can be constructed based on  the knowledge of $p$.

Additive-error channels have been studied in the literature~\cite{CN88,CN89,Lan08,GS16}.
In the previous studies,
the error rate is bounded, and the channel cannot see the transmitted codeword,
but can depend on the coding scheme.
In the present study, stronger obliviousness is considered, in which
the channel cannot depend on the code.

Samplable distributions were also studied in the context of 
data compression~\cite{GS91,TVZ05,Wee04}, 
randomness extractor~\cite{TV00,Vio11,DW12}, 
and randomness condenser~\cite{DRV12}.
Samplable distributions with membership test appeared in the study of efficient compressibility of samplable sources~\cite{GS91,TVZ05,Wee04}.

\subsection{Organization}
In Section~\ref{sec:pre}, we give the formal model of the problem, and introduce several notions of error-correcting codes.
Several results on the correctabilities of flat distributions are presented in Section~\ref{sec:flat}.
The negative result of correcting pseudorandom errors appears in Section~\ref{sec:pseudorandom}.
In Section~\ref{sec:membership}, we show the existence of an error distribution with membership test
that is not efficiently correctable.
The positive results are provided in Section~\ref{sec:positive}.

\section{Preliminaries}\label{sec:pre}


For $n \in \N$, we write $[n]$ as the set $\{1, 2, \dots, n\}$.
For a distribution $X$, we write $x \sim X$ to indicate that $x$ is chosen according to $X$.
We may use $X$ also as a random variable distributed according to $X$.
The \emph{support} of $X$ is $\Supp(X) = \{ x : \Pr_X(x) \neq 0\}$,
where $\Pr_X(x)$ is the probability that $X$ assigns to $x$. 
The \emph{Shannon entropy} of $X$ is given by
\textcolor{black}{$- \sum_{x \in \Supp(X)}\Pr_X(x) \log \Pr_X(x)$.}
For flat distributions, the Shannon entropy is equal to the \emph{min-entropy}, \textcolor{black}{which is defined
to be $\min_{x \in \Supp(X)}\{ - \log \Pr_X(x)\}$}.
Thus, we simply say that a flat distribution $Z$ has entropy $m$ if its Shannon entropy is $m$.
For $n \in \N$, we write $U_n$ as the uniform distribution over $\bin^n$.


We define the notion of additive-error correcting codes.

\begin{definition}[Additive-error correcting codes]
For two functions $\Enc : \F^{k} \to \F^n$ and $\Dec: \F^n \to \F^{k}$,
and a distribution $Z$ over $\F^n$, where \textcolor{black}{$k \leq n$ and} $\F$ is a finite field,
we say $(\Enc, \Dec)$ \emph{corrects (additive error) $Z$ with error  \textcolor{black}{probability} $\epsilon$}
if for any $x \in \F^{k}$, we have that
\begin{equation*}
  \Pr_{z \sim Z}[\Dec(\Enc(x) + z) \neq x] \leq \epsilon.
 \end{equation*}
\textcolor{black}{When $\epsilon = 0$, we simply say $(\Enc, \Dec)$ corrects $Z$.}
The \emph{rate} of $(\Enc, \Dec)$ is $k/n$.
\end{definition}

\begin{definition}
A distribution $Z$ is said to be \emph{correctable with rate $R$ and error \textcolor{black}{probability} $\epsilon$}
if there is a pair of functions $(\Enc, \Dec)$ of rate $R$ that corrects $Z$ with error \textcolor{black}{probability} $\epsilon$.
\end{definition}

We call a pair $(\Enc, \Dec)$ a \emph{coding scheme} or simply \emph{code}.
The coding scheme is called \emph{efficient} if $\Enc$ and $\Dec$ can be computed in polynomial-time in $n$.
The code is called \emph{linear}
if $\Enc$ is a linear mapping, 
that is, for any $x, y \in \F^n$ and $a, b \in \F$, $\Enc(ax+by) = a\,\Enc(x)+b\,\Enc(y)$.
If $|\F| = 2$, we may use $\bin$ instead of $\F$. 
For any linear code $(\Enc, \Dec)$, there is a matrix $G \in \F^{k \times n}$, called a
\emph{generator matrix}, such that $\Enc(x) = x \cdot G$ for all $x \in \F^k$.
We usually assume that $G$ has full rank, namely, the rank of $G$ is $k$.
  A matrix $H \in \F^{(n-k) \times n}$ is called a \emph{parity-check matrix} if
  for any $c \in \F^n$, $c = \Enc(x)$ for some $x \in \F^k$ if and only if $H c^T = 0$.
  Then, it holds that $GH^T =0$.
  It is well-known that given a parity-check matrix $H \in \F^{(n-k)\times n}$ of a code,
  one can compute a generator matrix $G$ of the code by finding a basis of the kernel of $H$.
  See~\cite{MWS78,Roth06} for the basic properties of linear codes.

Next, we define syndrome decoding for linear codes.
\textcolor{black}{Suppose $v = \Enc(x)+z \in \F^n$ is received,
  where $\Enc(x)$ is the transmitted codeword and $z$ is the error vector.
  Syndrome decoder associates $e$ with $v \cdot H^T$, called the \emph{syndrome} of $v$,
  because it holds that $v \cdot H^T = (\Enc(x) + z) \cdot H^T = x \cdot G \cdot H^T + z \cdot H^T = z \cdot H^T$.
  If there is a way for recovering $z$ from $z \cdot H^T$, we can recover $x$ from $v$.
}

\begin{definition}
For a linear code $(\Enc, \Dec)$,
$\Dec$ is said to be a \emph{syndrome decoder}
if there is a function $\Rec : \F^{\textcolor{black}{n-k}} \to \F^n$ such that
$\Dec(\textcolor{black}{v}) = (\textcolor{black}{v} - \Rec(\textcolor{black}{v} \cdot H^T)) \cdot G^{-1}$,
where
$G$ and $H$ are a generator matrix and a parity-check matrix of the code, respectively, and
$G^{-1} \in \F^{n \times \textcolor{black}{k}}$ is a right inverse matrix of $G$ (i.e., $G G^{-1}=I$).
\end{definition}
Note that for any full-rank matrix $G$, there always exist a right inverse matrix of $G$.
In the above definition, $\Rec$ recovers the error vector $e$ from the syndrome $v \cdot H^T$. Since $v - e = \Enc(x)$, $x$ is obtained by multiplying $G^{-1}$ by $\Enc(x) = x \cdot G$.

We consider a computationally-bounded analogue of additive-error channels.
We introduce the notion of samplable distributions.

\begin{definition}
A distribution family $Z = \{Z_n\}_{n \in \N}$ is said to be \emph{samplable} if
there is a probabilistic polynomial-time algorithm $S$ such that
$S(1^n)$ is distributed according to $Z_n$ for every $n \in \N$,
\textcolor{black}{where $1^n$ is the string consisting of $n$ ones.}
\end{definition}

We consider the setting in which coding schemes can depend on the sampling algorithm of $Z$,
but not on its random coins,
and $Z$ does not use any information on the coding scheme or transmitted codewords.
In this setting, the randomization of coding schemes does not help much.
\begin{proposition}\label{pro:deterministic}
Let $(\Enc, \Dec)$ be a randomized coding scheme
that corrects a distribution $Z$ with error \textcolor{black}{probability} $\epsilon$.
Then, there is a deterministic coding scheme that corrects $Z$ with error \textcolor{black}{probability} $\epsilon$.
\end{proposition}
\begin{proof}
Assume that $\Enc$ uses at most $\ell$-bit randomness.
Since $(\Enc, \Dec)$ corrects $Z$ with error \textcolor{black}{probability} $\epsilon$,
we have that for every $x \in \F^k$, $\Pr_{z \sim Z,r \sim U_\ell}[\Dec(\Enc(x;r)+z) \neq x] \leq \epsilon.$
By the averaging argument, for every $x \in \F^k$, there exists $r_x \in \bin^\ell$ such that
$\Pr_{z \sim Z}[\Dec(\Enc(x;r_x) + z)\neq x] \leq \epsilon.$
Thus, by defining $\Enc'(x) = \Enc(x;r_x)$,
the deterministic coding scheme $(\Enc', \Dec)$ corrects $Z$ with error \textcolor{black}{probability} $\epsilon$.
\end{proof}

\textcolor{black}{The above result reveals the crucial difference between our setting and that of
  Guruswami and Smith~\cite{GS16},
  where the channels can use the information on the coding scheme, but not transmitted codewords.}
They present a randomized coding scheme with optimal rate $1-h_2(p)$ for worst-case additive-error channels,
for which deterministic coding schemes are only known to achieve rate $1-h_2(2p)$ \textcolor{black}{for $p < 1/2$},
where $p$ is the error rate of the channels.

\ignore{
Next, we define the notion of data compression.
\begin{definition}
For two functions $\Com : \F^* \to \F^*$ and $\Decom: \F^* \to \F^*$, 
and a distribution $Z$,
we say $(\Com, \Decom)$ \emph{compresses} $Z$ to length $m$ if
\begin{enumerate}
\item For any $z \in \Supp(Z)$, $\Decom( \Com(z) ) = z$, and
\item $E[ | \Com(Z) |] \leq m$.
\end{enumerate}
\end{definition}

\begin{definition}
We say a distribution $Z$ is \emph{compressible} to length $m$,
if there are two functions $\Com$ and $\Decom$ such that
$(\Com, \Decom)$ compresses $Z$ to length $m$.
\end{definition}

If $\Com$ is a linear mapping, $(\Com, \Decom)$ is called a \emph{linear} compression.

Finally, we define the notion of lossless condensers.
\begin{definition}
A function $f : \F^n \times \bin^d \to \F^r$ is said to be an $(m, \epsilon)$-\emph{lossless condenser}
if for any distribution $X$ of min-entropy $m$,
the distribution $(f(X,Y),Y)$ is $\epsilon$-close to a distribution $(Z, U_d)$ with min-entropy at least $m+d$,
where $Y$ is the uniform distribution over $\bin^d$.
A condenser $f$ is \emph{linear} if for any fixed $z \in \bin^d$, any $x, y \in \F^n$ and $a, b \in \F$,
$f(ax+by,z) = af(x,z)+bf(y,z)$.
\end{definition}
} 

\section{Errors from Flat Distributions}\label{sec:flat}

We investigate the correctability of general flat distributions over $\bin^n$.



First, we show that, for any flat distribution $Z$ of entropy $m$, a random linear code of length $n$ and rate $R$ can correct $Z$ with error probability $1 - 2^{-\Omega(n)}$
for $R < 1 - m/n - o(1)$.
Consider a random linear code of length $n$ and rate $R$ such that
the parity check matrix $H$ is chosen uniformly at random from $\mathcal{H}_{n,R} = \bin^{(n-Rn)\times n}$,
\textcolor{black}{and a generator matrix $G \in \bin^{Rn \times n}$
  is obtained by finding a basis of the kernel of $H$.
  We use the syndrome decoding. Specifically,
  for a received word $v$, the function $\Rec$ of the decoder is defined
  such that it maps $v \cdot H^T$ to unique $z \in \Supp(Z)$ satisfying $v \cdot H^T = z \cdot H^T$ 
  by searching for all possible $z$. If there are multiple candidates for $z$,
  it outputs the decoding failure.}
Let $\mathcal{C}_{n,R}$ be the set of codes in which
each code is defined by using each element in $\mathcal{H}_{R,n}$ as a parity-check matrix.



\begin{proposition}\label{prop:rand}
For any flat distribution $Z$ over $\bin^n$ of entropy $m$,
a $(1-\sqrt{\epsilon})$-fraction of codes in $\mathcal{C}_{n,R}$ corrects $Z$ with error probability $\sqrt{\epsilon}$ 
\textcolor{black}{for $R < 1 - m/n$}, where $\epsilon = 2^{-n(1-R-m/n)}$.
\end{proposition}
\begin{proof}
  The probability that a random code from $\mathcal{C}_{n,R}$ fails the decoding is 
  \begin{align}
    & \Pr_{H \in \mathcal{H}_{n,R}, z \sim Z}\left[ \exists z' \in \Supp(Z) \setminus \{z\} : z \cdot H^T = z' \cdot H^T \right] \nonumber\\
    & = \Pr_{H, z} \left[ \exists z' \!\in \Supp(Z) \setminus \{z\} : 
        \forall i \in [n-Rn],
        h_i\cdot(z - z') = 0
      \right] \label{eq:2}\\
    & = \!\! \!\!\sum_{z \in \Supp(Z)} \!\!\!\!\!\! \Pr[z \sim Z]  \Pr_{H}\! \left[ \!\bigcup_{z' \in \Supp(Z) \setminus \{z\}}\!\!\left[
          \forall i \in [n-Rn],
          h_i\cdot(z - z') = 0
        \right] \!\right] \nonumber \\
    & \leq \!\! \sum_{z \in \Supp(Z)} \! \! \Pr[z \sim Z] \! \! \sum_{z' \in \Supp(Z) \setminus \{z\}} \!\! \Pr_H \left[
      \forall i \in [n-Rn], 
      h_i\cdot(z - z') = 0 
      \right] \nonumber\\
    & = \sum_{z \in \Supp(Z)} 2^{-m} \cdot (2^m-1) \cdot \left( \frac{1}{2} \right)^{n-Rn} \label{eq:3}\\ 
    & \leq \epsilon,\label{eq:4}
\end{align}
  where $H^T = (h_1^T, \dots, h_{n-Rn}^T)$ in (\ref{eq:2}),
  the first inequality follows from a union bound, and
  (\ref{eq:3}) follows from the facts that for any non-zero $a \in \bin^n$ and random $h \in \bin^n$, $\Pr_h[ h \cdot a = 0] = 1/2$,
  and that $|\Supp(Z)| = 2^m$ for flat $Z$.
Then, a $(1 - \sqrt{\epsilon})$-fraction of codes in $\mathcal{C}_{n,R}$ can correct $Z$ with error probability $\sqrt{\epsilon}$,
since otherwise (\ref{eq:4}) does not hold.
Therefore, the statement follows.
\end{proof}

\ignore{
For each $z \in \Supp(Z)$, the probability that a random code from $\mathcal{C}_{n,R}$ has the same syndrome for other $z' \in \Supp(Z)$ is
\begin{align}
& \Pr_{H \in \mathcal{H}_{n,R}}\left[ \exists z' \in \Supp(Z) \setminus \{z\} : z \cdot H^T = z' \cdot H^T \right] \nonumber\\
  & = \!\Pr_{H \in \mathcal{H}_{n,R}} \!\left[ \exists z' \!\in \Supp(Z) \setminus \{z\} : 
    \begin{aligned} &\forall i \in [n-Rn], \\& h_i\cdot(z - z') = 0 \end{aligned} \!\right] \label{eq:2}\\
& \leq \sum_{z' \in \Supp(Z) \setminus \{z\}} \prod_{i \in [n-Rn]} \Pr_{h_i \in \bin^n}\left[ h_i\cdot(z - z') = 0 \right] \nonumber\\
& \leq 2^m \cdot (1/2)^{n-Rn} = 2^{-n(1-R-m/n)},\nonumber
\end{align}
where $H^T = (h_1^T, \dots, h_{n-Rn}^T)$ in (\ref{eq:2}),
the first inequality follows from a union bound, and
the last inequality follows from the fact that $z-z' \neq 0$ and $|\Supp(Z)| = 2^{m}$ for flat $Z$.
The error probability that a code defined by $H \in \mathcal{H}_{n,R}$ fails the decoding is
\begin{align*}
& \Pr_{z \sim Z}\left[ \exists z' \in \Supp(Z) \setminus \{z\} : z \cdot H^T = z' \cdot H^T \right] 
\end{align*}

Thus, a random linear code has a unique syndrome for every $z \in \Supp(Z)$ is at least $1-2^{-n(1-R-m/n)}$,
and such a code can correct every error from $Z$.
}

Thus, we have the following proposition.
\begin{proposition}\label{prop:random}
Let $Z$ be any flat distribution over $\bin^n$ of entropy $m$.
There is a linear code of rate $R$ that corrects $Z$  with error probability $\epsilon$
for $R = 1 - m/n - 2\log(\epsilon^{-1})/n$. 
The decoding complexity is at most $O(n^22^m)$.
\end{proposition}
\begin{proof}
The existence of such a code immediately follows from Proposition~\ref{prop:rand}.
Given a received word \textcolor{black}{$v$}, the brute-force decoder checks if $(\textcolor{black}{v} - z) \cdot H^T = 0$ for all $z \in \Supp(Z)$,
where $H$ is the parity check matrix. If so, output $x$ satisfying $x \cdot G = \textcolor{black}{v} - z$. Thus, the decoding is done in time $O(n^2) \cdot |\Supp(Z)|$.
\end{proof}

Proposition~\ref{prop:random} implies that for any flat $Z$ of entropy $O(\log n)$,
there is a code that corrects $Z$ with error probability $\epsilon$ in polynomial time.
Although the construction is not fully explicit, we can obtain such a code with probability $1 - \epsilon$. 


Conversely, we can show that the rate achieved in Proposition~\ref{prop:random} is almost optimal.

\begin{proposition}\label{prop:flatbound}
Let $Z$ be any flat distribution over $\bin^n$ of entropy $m$.
If a code of rate $R$ corrects $Z$ with error \textcolor{black}{probability} $\epsilon$,
then $R \leq 1 - m/n \textcolor{black}{-} \log(1-\epsilon)/n$.
\end{proposition}
\begin{proof}
Let $(\Enc, \Dec)$ be a code that corrects $Z$ with error \textcolor{black}{probability} $\epsilon$.
For $x \in \bin^{Rn}$, define $D_x = \{ \textcolor{black}{v} \in \bin^n : \Dec(\textcolor{black}{v}) = x\}$.
That is, $D_x$ is the set of inputs that are decoded to $x$ by $\Dec$.
Since the code corrects the flat distribution $Z$ with error \textcolor{black}{probability} $\epsilon$,
$|D_x| \geq (1-\epsilon)2^m$ for every $x \in \bin^{Rn}$.
Since each $D_x$ is disjoint, $\sum_{x \in \bin^{Rn}}|D_x| \leq 2^n$.
Therefore, we have that $(1-\epsilon)2^m\cdot 2^{Rn} \leq 2^n$,
which implies the statement. 
\end{proof}

By Proposition~\ref{prop:random}, 
one may hope to construct a \emph{single} code that corrects
errors from any flat distribution with the same entropy,
as constructed in~\cite{Che09} for the case of binary symmetric channels
by using Justesen's construction~\cite{Jus72}.
However, it is impossible to construct such codes.
We show that for every deterministic coding scheme of rate $k/n$,
there is a flat distribution $Z$ of entropy $m$ that is not correctable by the scheme for any $1 \leq m \leq k$.
By combining this result with Proposition~\ref{pro:deterministic},
we can conclude that there is no coding scheme of rate $k/n$ that corrects every flat distribution of entropy $m$ with $1 \leq m \leq k$.
For deterministic coding schemes, we assume that the encoding function $\Enc$ is injective.
  Namely, for any distinct $x, x' \in \F^k,$ $\Enc(x) \neq \Enc(x')$.
  The assumption is necessary because otherwise the code always fails to decode 
  either $x$ or $x'$ such that $\Enc(x) = \Enc(x')$.


\begin{proposition}\label{prop:linear}
For any deterministic code of rate $k/n$ and any $m$ with $1 \leq m \leq k$,
there is a flat distribution of entropy $m$ that is not correctable by the code with error \textcolor{black}{probability} $\epsilon < 1/2$.
\end{proposition}
\begin{proof}
  \textcolor{black}{By assumption, the code contains $2^k$ codewords that are all distinct.}
Define a flat distribution to be a uniform distribution over 
$2^m$ distinct codewords $c_1, \dots, c_{2^m}$.
If the input to the decoder is $c_i + c_j$ for $i, j \in [2^m]$, the decoder cannot distinguish the two cases 
where the transmitted codewords are $c_i$ and $c_j$.
Thus, the decoder outputs the wrong answer with probability at least $1/2$ for at least one of the two cases.
\end{proof}

\section{Pseudorandom Errors}\label{sec:pseudorandom}

We show that 
no efficient coding scheme can correct pseudorandom errors,
\textcolor{black}{which can be sampled by pseudorandom generators.}

\begin{proposition}\label{prop:prg}
Assume that a one-way function exists.
Then, for any constant $\epsilon \in (0,1)$,
there is a samplable distribution $Z$ over $\bin^n$ of entropy
$n^\epsilon$ such that no polynomial-time algorithms $(\Enc, \Dec)$ can correct $Z$.
\end{proposition}
\begin{proof}
If a one-way function exists,
there is a pseudorandom generator $G : \bin^{n^\epsilon} \to \bin^n$ secure for any polynomial-time algorithm~\cite{HILL99}.
\textcolor{black}{Namely, the distribution $G(U_{n^\epsilon})$ is indistinguishable from the uniform distribution $U_n$ for any polynomial-time algorithm.}
Then, a distribution $Z = G(U_{n^\epsilon})$ is not correctable by polynomial-time algorithms $(\Enc, \Dec)$.
If so, we can construct a polynomial-time distinguisher for pseudorandom generator by employing $(\Enc, \Dec)$,
and thus a contradiction follows.
\end{proof}


\section{Errors with Membership Test}\label{sec:membership}

\ignore{
To derive this result, we use the technique of Wee~\cite{Wee04},
which is based on the \emph{reconstruction paradigm} of Gennaro and Trevisan~\cite{GGKT05}.
We use his technique for the problem of error correction.
We show that if a samplable distribution with a sampler $S$ is efficiently correctable,
then the function of $S$ has a short description, and thus, by a counting argument, 
efficient coding schemes cannot correct every samplable distribution with membership test.

This negative result seems counterintuitive.
In general, constructing low-rate codes seems to be easier than high-rate codes.
However, the result implies the impossibility of constructing low-rate codes.
The reason for such a result is that 
the reconstruction paradigm crucially uses the fact that some function can be described shortly.
In our case, we use the fact that functions for samplable errors can be described shortly
if the errors are correctable by coding schemes with short descriptions.
Since low-rate codes have short descriptions, the result can be applied to low-rate codes.
-------------------------
}

In this section, we present samplable distributions that are not pseudorandom,
but cannot be corrected by efficient coding schemes. 
For such distributions, we consider distributions for which the membership test can be done efficiently.
A distribution $Z$ is called a \emph{distribution with membership test}
if there is a polynomial-time algorithm $D$ such that
$D(z) = 1 \Leftrightarrow z \in \Supp(Z)$.
Since the algorithm $D$ can distinguish $Z$ from the uniform distribution,
$Z$ is not pseudorandom.

\textcolor{black}{We consider a sampling algorithm/circuit that can access an \emph{oracle},
  which, on querying input $x$, responds with the value $f(x)$ for an a priori specified function $f$.
  It is assumed that the algorithm can obtain the value of any input in a single step,
  and that other algorithms can also access the same oracle.
  Such oracle algorithms/circuits are used in the field of computational complexity~\cite{AB09}.
  It is said to be relative to an oracle if algorithms are allowed to access it.}

We show that there is an oracle relative to which there exists a samplable distribution with membership test
that is not correctable by efficient coding schemes with low rate.
\textcolor{black}{In the proof, we use the technique called \emph{reconstruction paradigm}~\cite{GGKT05,Wee04}.}
Before starting the proof, we briefly describe the technique of~\cite{GGKT05,Wee04} and our proof idea.

In~\cite{GGKT05}, the paradigm is used to prove that a random function is one-way with high probability.
  Roughly speaking, it is shown that if $f$ is not one-way, then $f$ has a ``short'' description.
  Here, we say $f$ has a short description if, given access to some oracle $A$,
  the truth table of $f$ can be reconstructed by using $A$ and some short information.
  It is shown in~\cite{GGKT05} that there are not so many functions with short description,
  and thus a random function is one-way with high probability.
  The technique was used in~\cite{Wee04} to show the existence of incompressible samplable source with low pseudoentropy.


Here, we use the technique to prove the existence of uncorrectable
  samplable distributions with membership test.
  By following the approach of~\cite{Wee04}, we define a class of functions $\correctf$,
  which contains functions $f$ that can be efficiently corrected by some coding scheme.
  Then, we show that $\correctf$ has a short description.
  Since there are not so many functions with short description,
  we can show that there is a function $f$ that cannot be corrected with efficient coding schemes.

We now begin the formal proof.
Let $N = 2^n, K = 2^k, M = 2^m$.
Let $\mathcal{F}$ be the set of injective functions $f : \bin^m \to \bin^n$.
For each $f \in \calF$, define an oracle $\calO_f$ such that
\begin{align*}
 \calO_f(b,y)  & =
\begin{cases}
\calO_f^S(y) & \text{if } b = 0, y \in \bin^m\\
\calO_f^M(y) & \text{if } b = 1, y \in \bin^n\\
\bot & \text{otherwise}
\end{cases},\\
 \calO^M_f(y)  & = 
\begin{cases}
1 & \text{if } y \in f(\bin^m)\\
0 & \text{if } y \notin f(\bin^m)\\
\end{cases}, \\
 \calO^S_f(y)  & = f(y).
\end{align*}
\textcolor{black}{Namely, on input $y$, the oracle $\calO^M_f$ outputs the membership of $y$ under $f$,
  and $\calO^S_f$ samples $f(y)$.
  Thus, sampling $f(y)$ can be performed by querying $(0,y)$ to $\calO_f$, and membership test of $y$ can be done by $(1,y)$.}

Let $\correctf$ be the set of functions $f \in \calF$ for which
there exist oracle circuits $(\Enc, \Dec)$ that make $q$ queries to oracle $\calO_f$ \textcolor{black}{in total}
and correct $f(U_m)$ with rate $k/n$.
For each
\textcolor{black}{$f \in \correctf$, fix a pair $(\Enc, \Dec)$ that make $q$ queries to $\calO_f$ and correct $f(U_m)$ with rate $k/n$. We}
define
\begin{align*}
\invert_f & = \{ y \in \bin^m : \text{for any $x \in \bin^k$, on input} 
\text{$\Enc(x) + f(y)$, $\Dec$ queries $\calO_f^S$ on $y$} \}, \\
\forge_f & = \{ y \in \bin^m : \text{for some $x \in \bin^k$, on input} 
\text{$\Enc(x) + f(y)$, $\Dec$ does not query $\calO_f^S$ on $y$} \}.
\end{align*}
Note that $\invert_f$ and $\forge_f$ is a partition of $\bin^m$.
We also define
\begin{align*}
\invertible & = \{ f \in \correctf : |\invert_f| > \epsilon \cdot 2^m \},\\
\forgeable & = \{ f \in \correctf : |\forge_f| \geq \delta \cdot 2^m \},
\end{align*}
where $\epsilon$ and $\delta$ are any positive constants satisfying $\epsilon + \delta = 1$.
Note that $\correctf = \invertible \cup \forgeable$.
\textcolor{black}{Roughly speaking,
  if $f \in \invertible$, an $\epsilon$-fraction of $f(y)$ can be corrected by querying $\calO_f^S$ on $y$,
  which implies that $f(y)$ is invertible by $\Dec$.
  For $f \in \forgeable$, a $\delta$-fraction of $f(y)$ can be corrected without querying $\calO_f^S$ on $y$,
  which implies that $f(y)$ is generated by $\Dec$ with no help from the oracle.
}

\textcolor{black}{
  In the following, we show that every $f \in \correctf = \invertible \cup \forgeable$ has a short description.
  For each $f \in \correctf$, we present a way for constructing the truth table of $f$ by employing $(\Enc, \Dec)$.
  If $f \in \invertible$, it}
is done by computing $\Enc(x) + f(y)$ and monitoring oracle queries that
$\Dec(\Enc(x)+f(y))$ makes to $\calO_f^S$.
\textcolor{black}{
  For $y$ on which $\Dec$ does not query $\calO_f^S$, the pair $(y, f(y))$ is stored in a look-up table.
  }
Similarly, if $f \textcolor{black}{ \, \in \,} \forgeable$, then $\Dec$ corrects $f(y)$ without querying $\calO_f^S$ on $y$.
This means that 
$f(y)$ can be described using $\Dec$ and $\Enc(x)+f(y)$, and thus if $\Enc(x)+f(y)$ has a short description,
the size of $\forgeable$ is small.


First, we show that $f \in \invertible$ has a short description.
\begin{lemma}\label{lem:invertible}
Take any $f \in \invertible$ and \textcolor{black}{a} pair of oracle circuits $(\Enc, \Dec)$ that makes at most $q$ queries to $\calO_f$ in total
and corrects $f(U_m)$ with rate $k/n$.
Then $f$ can be described using at most
\[ \log{N \choose c} + \log{M \choose c} + \log\left( {N-c \choose M-c}(M-c)!\right) \]
bits, given $(\Enc, \Dec)$, where $c = \epsilon M/q$.
\end{lemma}
\begin{proof}
First, consider an oracle circuit $A$ such that, on input $z$, 
$A$ picks any $x \in \bin^k$ and simulates $\Dec$ on input $\Enc(x) + z$.
Then, for any $y \in \invert_f$, on input $f(y)$, $A$ outputs $y$ by making at most $q$ queries to $\calO_f$.


Next, we show that for any $f \in \invertible$, $f$ has a short description given $A$.
Without loss of generality, we assume that $A$ makes distinct queries to $\calO_f^S$.
We also assume that on input $f(y)$, $A$ always queries $\calO_f^S$ on $y$ before it outputs $y$.
We will show that there is a subset $T \subseteq f(\invert_f)$ such that $f$ can be described given
$T$, $B(T)$, $f|_{\bin^m \setminus B(T)}$, where \textcolor{black}{$f|_X$ denotes the set $\{ (x, f(x)) : x \in X\}$} and $B(T) = \{ y \in \bin^m : y \leftarrow A(z), z \in T\}$.

We describe how to construct $T$ below.

\medskip
\noindent \textsc{Construct-$T$}:
\begin{enumerate}
\magicwand

\item Initially, $T$ is empty, and all elements in $T^* = f(\invert_f)$ are candidates for inclusion in $T$.

\item \label{step:2**} Choose the lexicographically smallest $z$ from $T^*$, put $z$ in $T$, and remove $z$ from $T^*$.

\item \label{step:3*} Simulate $A$ on input $z$, and halt the simulation immediately after $A$ queries $\calO_f^S$ on $y$.
  Let $y_1', \dots, y_p'$ be the queries that $A$ makes to $\calO_f^S$, where $y_p' = y$ and $p \leq q$.
\begin{itemize}
\magicwand
\item Remove $f(y_1'), \dots, f(y_{p-1}')$ from $T^*$.
(This means that these elements will never belong to $T$, and in simulating $A(z)$ in the recovering phase,
the answers to these queries are made by using the look-up table for $f$.)

\item Continue to remove the lexicographically smallest $z$ from $T^*$ 
until we have removed exactly $q-1$ elements in Step~\ref{step:3*}.
\end{itemize}

\item Return to Step~\ref{step:2**}.
\end{enumerate}

Next, we describe how to reconstruct $f$ from $T$, $B(T)$, and $f|_{\bin^m \setminus B(T)}$.
We show how to recover the look-up table for $f$ on values in $B(T)$.

\medskip
\noindent \textsc{Recover-$f$}:
\begin{enumerate}
\magicwand
\item \label{step:1*} Choose the lexicographically smallest element $z \in T$, and remove it from $T$.

\item \label{step:2*} Simulate $A$ on input $z$, and halt the simulation immediately after $A$ queries $\calO_f^S$ on $y$ for which
the answer does not exist in the look-up table for $f$.
Since the query $y$ satisfies that $y = f^{-1}(z)$, add the entry $(y, z)$ to the look-up table.

\item Return to Step~\ref{step:1*}.
\end{enumerate}

We explain why we can correctly simulate $A(z)$ in Step~\ref{step:2*} of \textsc{Recover-$f$}.
Since $B(T)$ and $f|_{\bin^m \setminus B(T)}$ are given, we can answer all queries to $O_f^M$.
For any query $y'$ to $O_f^S$, it must be either (1) $y' \notin B(T)$, or (2) $y'$ is the output of $A$ on input $z'$ such that
$z' \in W$ and $z'$ is lexicographically smaller than $z$.
In either case, the look-up table has the corresponding entry, and thus we can answer the query.

In each iteration in \textsc{Construct-$T$}, we add one element to $T$ and remove exactly $q$ element from $T^*$.
Since initially the size of $T^* = f(\invert_f)$ is $\epsilon M$, the size of $T$ in the end is $c = \epsilon M/q $.

Finally, we evaluate the number of bits to describe $T$, $B(T)$, and $f|_{\bin^m \setminus B(T)}$.
Since $T$ can be specified by choosing $c$ elements from the set $f(\invert_f) \subseteq \bin^n$ of size at most $N$,
$T$ can be described using at most $\log{N \choose c}$ bits.
Similarly, $B(T)$ can be specified by choosing at most $|T|$ elements from the set $\bin^m$ of size $M$.
Hence, $B(T)$ can be described using $\log{M \choose c}$ bits.
The look-up table for $f|_{\bin^m \setminus B(T)}$ consists of $\{ (y,f(y)) : y \in \bin^m \setminus B(T)\}$,
which can be specified by first choosing the elements of $\{f(y) : y \in \bin^m \setminus B(T)\}$ from $\bin^n \setminus f(B(T))$
and then ordering them lexicographically with respect to the input $y \in \bin^m \setminus B(T)$.
Since $f$ is injective, we have that $|B(T)| = |T| = c$, and thus  $|\{f(y) : y \in \bin^m \setminus B(T)\}| = M-c$,
$|\bin^n \setminus f(B(T))| = N - c$, and $|\bin^m \setminus B(T)| = M -c$.
Thus, the look-up table for $f|_{\bin^m \setminus B(T)}$ can be described using $\log( {N-c \choose M-c}(M-c)!)$ bits.
Therefore, the statement follows.
\end{proof}

We show that the fraction of $f \in \calF$ for which $f \in \invertible$ and $f(U_m)$ is correctable is small.
\begin{lemma}\label{lem:invertiblecount}
If $m > 3 \log s + \log n + O(1)$,
then the fraction of functions $f \in \calF$ such that $f \in \invertible$ and 
$f(U_m)$ can be corrected by a pair of oracle circuits $(\Enc, \Dec)$ of total size $s$
is less than $2^{-(sn\log s+1)}$
for all sufficiently large $n$.
\end{lemma}
\begin{proof}
It follows from Lemma~\ref{lem:invertible} that, given $(\Enc, \Dec)$, the fraction is 
\begin{equation*}
\frac{|\invertible|}{{N \choose M}M!}  \leq \frac{{N \choose c}{M \choose c}{N-c \choose M-c}(M-c)!}{{N \choose M}M!}
 = \frac{{M \choose c}}{c!},
\end{equation*}
where $c = \epsilon M/(qK)$.
By using the fact that $q \leq s$ and the inequalities ${n \choose k} < \left( \frac{en}{k} \right)^k$ and $n! > \left( \frac{n}{e} \right)^n$,
the expression is upper bounded by
\begin{equation*}
\left(\frac{eM}{c} \right)^c \left( \frac{e}{c} \right)^c\\
 = \left( \frac{e^2 q^2}{\epsilon^2M} \right)^{{\epsilon M}/{q}}\\
< \left( \frac{1}{2} \right)^{ns \log s + 1}
\end{equation*}
for all sufficiently large $n$.
The last inequality follows from the fact that 
\begin{align*}
\frac{e^2 q^2}{\epsilon^2 M}  < \frac{e^2 q^2}{\epsilon^2 \, \Omega(s^3 n)} < \frac{1}{2} \text{\quad and \quad} 
\frac{\epsilon M}{q} > \frac{\epsilon \, \Omega(s^3 n)}{q} > n s \log s + 1.
\end{align*}
\end{proof}

Next, we show that $\forgeable$ has a short description.
\begin{lemma}\label{lem:forgeable}
Take any $f \in \forgeable$ and \textcolor{black}{a} pair of oracle circuits $(\Enc, \Dec)$ that make at most $q$ queries to $\calO_f$ in total
and corrects $f(U_m)$ with rate $k/n$.
Then $f$ can be described using at most 
\[ \log{M \choose d} + \log\left( {N-d \choose M-d}(M-d)!\right) + d(k + m + \log q)\]
bits, given $(\Enc, \Dec)$, where $d = \delta M/q$.
\end{lemma}

\begin{proof}
First, consider an oracle circuit $A$ such that, on input $w$,
$A$ obtains $x$ by simulating $\Dec$ on input $w$, 
queries $\calO_f^M$ on $w - \Enc(x)$,
and outputs $\bot$ if $\calO_f^M(w - \Enc(x)) = 0$, and $x$ otherwise.
Then, $A$ satisfies that, on input $w$, $A$ outputs $\bot$ if $w \notin \Enc(\bin^k)+f(\bin^m)$,
and $\Dec(w)$ otherwise.

Next, we show that for any $f \in \forgeable$, $f$ has a short description given $A$.
Without loss of generality, we assume that $A$ makes distinct queries to $\calO_f^S$ and $\calO_f^M$.
We also assume that 
for $x \in \bin^k$ and $y \in \bin^m$, $A(\Enc(x)+f(y))$ always queries $\calO_f^M$ on $f(y)$ before it outputs $x$.
Note that for $y \in \forge_f$, there is some $x \in \bin^k$ such that, on input $\Enc(x)+f(y)$, $A$ does not query $\calO_f^S$ on $y$.

We will show that there is a subset $Y \subseteq \forge_f$
such that $f$ can be described given
$Y$, $f|_{\bin^m \setminus Y}$, and $\{ (x_y, a_y, b_y) \in \bin^k \times [M] \times [q] : y \in Y\}$ of a set of advice strings.
For $x \in \bin^k$, we define $D(x) = \{ \Enc(x) + f(y) : y \in \bin^m \}$. 
Note that $|D(x)| = M$ for any $x \in \bin^k$.


We describe how to construct $Y$ below.

\medskip
\noindent \textsc{Construct-$Y$}:
\begin{enumerate}
\magicwand
\item Initially, $Y$ is empty. All elements in $Y^* = \forge_f$ are candidates for inclusion in $Y$.
For every $x \in \bin^k$, set $D_x = \{ \Enc(x) + f(y) : y \in \forge_f \}$. 
We write $\mathcal{D}_k = \bigcup_{x \in \bin^k}D_x$.

\item  \label{CYstep:2}
Choose the lexicographically smallest $y$ from $Y^*$, put $y$ in $Y$, and remove $y$ from $Y^*$.

\item
Choose the lexicographically smallest $w$ from 
the set of $\Enc(x) + f(y) \in D_x$ such that $A$ does not query $\calO_f^S$ on $y$.
If $w = \Enc(x) + f(y)$, set $x_y = x$.
Then, for every $x' \in \bin^k$, remove $\Enc(x')+f(y)$ from $D_{x'}$.
(This removal means that hereafter there are no elements in $\mathcal{D}_k$ for which $A$ outputs some $x$ such that $f(y)$ is the error vector.)
When $w$ is the lexicographically $t$-th smallest element in $D(x)$, set $a_y = t$
(so that we can recognize that the $a_y$-th element in $D(x)$ is $w$ in the recovering phase).


\item  \label{CYstep:4} Simulate $A$ on input $w$, and halt the simulation immediately after $A$ queries $\calO_f^M$ on $f(y)$.
Let $y_1', \dots, y'_p$ be the queries that $A$ makes to $\calO^S_f$, and 
$z_1', \dots, z_r'=f(y)$ be  the queries that $A$ makes to $\calO^M_f$. 
Set $b_y = r$
(so that we can recognize that the $b_y$-th query that $\Dec$ makes to $\calO_f^M$ is $f(y)$ in the recovering phase).

\begin{enumerate}
\magicwand
\item For every $x' \in \bin^k$,
remove $\Enc(x')+f(y_1'), \dots, \Enc(x')+f(y_p')$ from $D_{x'}$. 

\item For every $i \in [p]$, if $z_i' \in f(\forge_f)$, then for every $x' \in \bin^k$, remove $\Enc(x')+z_i'$ from $D_{x'}$, and otherwise, do nothing.

\item 
Continue to remove the elements $\Enc(x')+f(y)$ from $D_{x'}$ for every $x' \in \bin^k$ 
for the lexicographically smallest $w = \Enc(x)+f(y) \in \mathcal{D}_k$
until we have removed exactly $(q-1)K$ elements from $\mathcal{D}_k$ in Step~\ref{CYstep:4}.
\end{enumerate}

\item Return to Step~\ref{CYstep:2}.
\end{enumerate}

Next, we describe how to construct $f$ from $Y$, $f|_{\bin^m \setminus Y}$, and $\{ (x_y, a_y, b_y) \in \bin^k \times [M] \times [q] : y \in Y\}$.
We show how to recover the look-up table for $f$ on values in $Y$.

\medskip
\noindent \textsc{Recover-$f$:}

\begin{enumerate}
\magicwand

\item \label{Rfstep:1} 
Choose the lexicographically smallest $y \in Y$, and remove it from $Y$.
Then, choose the lexicographically $a_y$-th smallest element $w$ from $D(x_y)$.


\item \label{Rfstep:2}Simulate $A$ on input $w$, and halt the simulation immediately after $A$ makes the $b_y$-th query to $\calO_f^M$.
Since the $b_y$-th query is $f(y)$, add the entry $(y, f(y))$ to the look-up table.

\item Return to Step~\ref{Rfstep:1}.
\end{enumerate}

We explain why we can correctly simulate $A(w)$ in Step~\ref{Rfstep:2} of \textsc{Recover-$f$}.
For any query $y'$ to $\calO_f^S$, it must be either 
(1) $y' \notin Y$ or (2) $y'$ is lexicographically smaller than $y$.
In case (1), we can answer the query by using $f|_{\bin^m \setminus Y}$.
In case (2), since $y$ was chosen as the lexicographically smallest element such that $A$ does not query $\calO_f^S$ on $y$, 
the look-up table has the answer to the query.
Consider any of the first $b_y-1$ queries $z'$ to $\calO_f^M$.
If $z' \in f(\bin^m)$, namely $z' = f(y')$ for some $y'$,
then it must be either (1) $y' \notin Y$ or (2) $y'$ is lexicographically smaller than $y$. 
In either case, the look-up table has the entry $(y', z')$. 
If $z' \notin f(\bin^m)$, there is no entry for $z'$ in the look-up table.
Thus, we can answer the query by saying ``yes'' if $z'$ is in the look-up table, and ``no'' otherwise.

In each iteration in~\textsc{Construct-$Y$}, 
we add one element to $Y$ and remove exactly $qK$ elements from $\mathcal{D}_k$.
Since initially the size of $\mathcal{D}_k$ is at least $\delta KM$,
the size of $Y$ in the end is at least $d = \delta M/q$.

Finally, we evaluate the number of bits to describe $Y$, $f|_{\bin^m \setminus Y}$, and $\{ (x_y, a_y, b_y) \in \bin^k \times [M] \times [q] : y \in Y\}$.
Since $Y$ can be specified by choosing $d$ elements from the set $\forge_f \subseteq \bin^m$ of size at most $M$,
$Y$ can be described using $\log{ M \choose d}$ bits.
By the same argument in the proof of Lemma~\ref{lem:invertible} for $f|_{\bin^m \setminus B(T)}$,
we can show that the look-up table for $f|_{\bin^m \setminus Y}$ can be described using $\log( {N-d \choose M-d}(M-d)!)$ bits.
By simply listing the elements, the set $\{ (x_y, a_y, b_y) \in \bin^k \times [M] \times [q] : y \in Y\}$ can be described using $d (k + m +\log q)$ bits.
Therefore, the statement follows.
\end{proof}

We show that the fraction of $f \in \calF$ for which $f \in \forgeable$ and $f(U_m)$ is correctable is small.
\begin{lemma}\label{lem:forgeablecount}
If $m > 3 \log s + \log n + O(1)$ and $m < n-k-2\log s-O(1)$,
then the fraction of functions $f \in \calF$ such that $f \in \forgeable$ and 
$f(U_m)$ can be corrected by a pair of oracle circuits $(\Enc, \Dec)$ of total size $s$
is less than $2^{-(sn\log s+1)}$
for all sufficiently large $n$.
\end{lemma}
\begin{proof}
It follows from Lemma~\ref{lem:forgeable} that, given $(\Enc, \Dec)$, the fraction is 
\begin{align*}
  \frac{|\forgeable|}{{N \choose M}M!}  
  \leq \frac{{M \choose d}{N-d \choose M-d}(M-d)!}{{N \choose M}M!}2^{d(k+m+\log q)} 
  = \frac{{M \choose d}}{{N \choose d}d!}(qKM)^d,
\end{align*}
where $d = \delta M/q$.
By using the fact that $q \leq s$ and the inequalities ${n \choose k} < \left( \frac{en}{k} \right)^k$, ${n \choose k} > \left( \frac{n}{k} \right)^k$, and $n! > \left( \frac{n}{e} \right)^n$,
the expression is upper bounded by
\begin{align*}
 \left(\frac{eM}{d} \right)^d \left(\frac{d}{N} \right)^d \left( \frac{e}{d} \right)^d (qKM)^d
 = \left( \frac{e^2 q^2 KM}{\delta N} \right)^{\delta M/q}
 < \left( \frac{1}{2} \right)^{ns \log s +1}
\end{align*}
for all sufficiently large $n$.
The last inequality follows from the fact that 
\begin{align*}
  \frac{e^2 q^2 KM}{\delta N}  <  \frac{e^2 q^2}{\delta \, \Omega(s^2n)} < \frac{1}{2}
\text{\quad and \quad} 
\frac{\delta M}{q} > \frac{\delta \, \Omega(s^3 n)}{q} > n s \log s + 1.
\end{align*}
\end{proof}

We obtain the main result of this section.
\begin{theorem}\label{th:lowentropy}
For any $m$ and $k$ satisfying $3 \log s  + \log n + O(1) < m < n - k  - 2\log s - O(1)$,
there exist injective functions $f : \bin^m \to \bin^n$ such that,
given oracle access to $\calO_f$,
(1) $f(U_m)$ is a samplable distribution with membership test of entropy $m$,
and (2) $f(U_m)$ cannot be corrected with rate $k/n$ by oracle circuits of size $s$.
\end{theorem}
\begin{proof}
Since $\correctf = \invertible \cup \forgeable,$
it follows from Lemmas~\ref{lem:invertiblecount} and~\ref{lem:forgeablecount}
that for a fixed $(\Enc, \Dec)$ of size $s$,
the fraction of functions $f \in \calF$ such that $(\Enc, \Dec)$ corrects $f(U_m)$ with rate $k/n$
is less than $2^{-(sn\log s)}$.
Since there are at most $2^{sn \log s}$ circuits of size $s$, 
there are functions $f \in \calF$ such that $f(U_m)$ cannot be corrected with rate $k/n$ by oracle circuits of size $s$.
Given oracle access to $\calO_f$, $f(U_m)$ is samplable. Since $f$ is injective, $f(U_m)$ has entropy $m$.
\end{proof}

The following corollary immediately follows.
\begin{corollary}\label{cor:lowentropy}
For any $m$ and $k$ satisfying $\omega(\log n) < m < n - k - \omega(\log n)$,
there exists an oracle relative to which there exists a samplable distribution with membership test of entropy $m$
that cannot be corrected with rate $k/n$ by polynomial size circuits.
\end{corollary}



\section{Positive Results}\label{sec:positive}

In this section, we present positive results for correcting samplable additive errors.

\subsection{Errors from Linear Subspaces}

We show that if the set of error vectors forms a linear subspace,
every error can be corrected by a linear code with optimal rate.
Let $Z' = \{z_1, z_2, \dots, z_m\} \subseteq \F^n$ be a set of linearly independent vectors.
We construct a linear code that corrects additive errors from the linear span of $Z'$.

\begin{proposition}\label{prop:subspace}
  \textcolor{black}{Let $Z$ be the uniform distribution over the linear span of $Z'$,
  which has entropy $m$.}
There is a linear code of rate $1-m/n$ that corrects \textcolor{black}{$Z$}  by syndrome decoding.
\end{proposition}
\begin{proof}
Consider $n-m$ vectors $w_{m+1}, \dots, w_n \in \F^n$ such that the set $\{z_1, z_2, \dots, z_m, w_{m+1}, \dots, w_n\}$ forms a basis of $\F^n$.
Then, there is a linear transformation $T : \F^n \to \F^m$ such that $T(z_i) = e_i$ and $T(w_i) = 0$,
where $e_i$ is the vector with $1$ in the $i$-th position and $0$ elsewhere.
Let $H$ be the matrix in $\F^{m \times n}$ such that $x H^T = T(x)$,
and consider a code with parity check matrix $H$.
Let $z = \sum_{i=1}^m a_i z_i$ be a vector in the linear span of $Z'$, where $a_i \in \F$.
Since $z \cdot H^T = (\sum_{i=1}^m a_i z_i) \cdot H^T = \sum_{i=1}^m a_i e_i = (a_1, \ldots, a_m)$,
the code can correct the error $z$ by syndrome decoding.
Since $H \in \F^{m \times n}$ is the parity check matrix, the rate of the code is $(n- m)/n$.
\end{proof}

\subsection{A Computational Condition}\label{sec:compcond}

Let $Z = f(U_r)$ be a flat distribution over $\bin^n$ of entropy $m$ associated with a samplable additive channel,
where $f$ is an efficiently computable function, and $r \geq m$.
We give a computational condition under which $Z$ is efficiently correctable.
Roughly speaking, we show that if a variant of the function $f$ is efficiently ``invertible'',
then $Z$ can be efficiently corrected.

Specifically, for function $f : \bin^r \to \bin^n$,
we define function $g : \bin^r \times \calH \to \calH \times \bin^{m+2c \log n}$ such that
\begin{equation*}
  g(y, h) = (h, h(f(y))),
  \end{equation*}
where 
$\calH = \{h : \bin^n \to \bin^{m+2c\log n}\}$ is  a family of \emph{linear universal hash functions}~\cite{CW79,TH13},
and $c$ is a positive constant.
The universality means that for any distinct $x, x' \in \bin^n$,
\begin{equation*}
  \Pr_{h \in \calH}[ h(x) = h(x')] \leq 2^{-(m+2c\log n)},
\end{equation*}
and the linearity means that for any $x, x' \in \bin^n$ and $a, b \in \bin$, $h(ax+bx') = ah(x) + bh(x')$.

As efficient ``invertibility'', we introduce the notion of \emph{distributionally one-way function}~\cite{IL89}.
Intuitively, such a function $g$ guarantees that the distribution $(x, g(x))$ for random $x$ is difficult to be simulated by efficient algorithms given only $g(x)$.
\begin{definition}
A function $g$ is said to be \emph{distributionally one-way} if it is computable in polynomial time
and there exists a constant $c > 0$ such that
for every probabilistic polynomial-time algorithm $A$, the statistical distance between
$(x, g(x))$ and $(A(g(x)), g(x))$ is at least $1/n^c$, where $x \sim U_n$.
\end{definition}

We show that if the function $g$ defined above is not distributionally one-way,
then the error distribution $Z = f(U_r)$ is efficiently correctable.
Before giving the formal proof, we describe the underlying idea.

We employ the technique used in the proof of~\cite[Theorem~6.3]{Wee04} that shows the necessity of one-way functions 
for separating pseudoentropy and compressibility.
In the proof, it is observed that every samplable distribution $Z$ can be
optimally compressed by a hash function $h$ such that a sample $z$ from $Z$ is simply compressed to $h(z)$.
In addition, if distributionally one-way functions do not exist, it is shown that $z$ can be recovered from $h(z)$
by a polynomial-time algorithm.
We observe that 
a family of \emph{linear} hash functions is used for giving a recovering algorithm.
Since a linear compression function is a dual object of a linear code that corrects additive errors~\textcolor{black}{\cite{CSV04}},
we can construct a linear code correcting additive errors of $Z$.
More specifically, we construct an efficient syndrome decoder that recovers $z = f(y)$ from a syndrome $h(z)$
by assuming that $f$ is not distributionally one-way.
A problem for recovering $z$ from $h(z)$ is that there may be exponentially many preimages of $h(z)$,
and we need to choose $z$ that is in the support of $Z$.
To solve the problem, we define the function $g(y, h) = (h, h(f(y)))$,
and construct a code by assuming that $g$ is not distributionally one-way.

\begin{theorem}\label{thm:distow}
  If $g(y,h) = (h,h(f(y)))$ is not distributionally one-way,
  then the flat distribution $Z = f(U_r)$ over $\bin^n$ of entropy $m$ can be
  corrected by a polynomial-time coding scheme of rate $1 - m/n - (2c\log n)/n$ and error \textcolor{black}{probability} $O(n^{-c})$.
\end{theorem}
\begin{proof}
  For each $h \in \calH$, define $C_h = \{ z \in \Supp(Z) : \exists z' \in \Supp(Z) \text{ s.t. } z' \neq z \wedge h(z) = h(z') \}$.
  Namely, $C_h$ is the set of inputs with collisions under $h$.
  By a union bound, it holds that for any $z \in \Supp(Z)$,
  \begin{align*}
    \Pr_{h \in \calH}[ \exists z' \in \Supp(Z) : z' \neq z \wedge h(z') = h(z)] 
    \leq \frac{|\Supp(Z)|}{2^{m+2c\log n}} \leq \frac{1}{n^{2c}}.
  \end{align*}
  Thus, ${\mathbb{E}}_{h \in \calH}\left[ |C_h| \right] \leq {2^m}/{n^{2c}}.$
We say $h \in \calH$ is good if $|C_h| \leq 2^m/n^c$. 
By Markov's inequality, we have that
\begin{equation*}
  \Pr_{h \in \calH}\left[ |C_h| > \frac{2^m}{n^c}\right] < \frac{1}{n^c}.
\end{equation*}

By the assumption that $g$ is not distributionally one-way,
there is a polynomial-time algorithm $A$ such that the statistical distance between
$((y, h), g(y, h))$ and $(A(g(y,h)), g(y,h))$ is at most $n^{-c}$, where $y \sim U_r$ and $h \in \calH$.
Then, it holds that
\begin{equation*}
\Pr_{A,y,h}[g(A(g(y,h))) = g(y,h)] \geq 1 - \frac{1}{n^c},
\end{equation*}
where the probability is taken over the random coins of $A$, $y \sim U_r$, and $h \in \calH$.
Thus, we have that
\begin{equation*}
\Pr_{A,y,h}[ g(A(g(y,h))) = g(y,h) \wedge \text{$h$ is good}] \geq 1 - \frac{2}{n^c}.
\end{equation*}
By fixing the coins of $A$ and $h \in \calH$,
it holds that there are deterministic algorithm $A'$ and $h_0 \in \calH$ such that $h_0$ is good and
\begin{equation*}
\Pr_{y \sim U_r}[ g(A'(g(y,h_0))) = g(y,h_0) ] \geq 1 - \frac{2}{n^c}.
\end{equation*}
For $y \in \bin^r$ satisfying $g(A'(g(y,h_0))) = g(y, h_0)$, 
we write $A'(g(y,h_0)) = (A_1'(g(y,h_0)), A_2'(g(y,h_0))) = (y', h')$.
Then, it holds that $h' = h_0$ and $h_0(f(y)) = h_0(f(y'))$.
Furthermore, since $h_0$ is good, $\Pr_y[ f(y) \notin C_{h_0}] \geq 1 - 1/n^c$.
Since $\calH$ is a set of linear hash functions, there is a matrix
$H_0 \in \bin^{(m + 2c\log n) \times n}$ such that $x H_0^T = h_0(x)$ for $x \in \bin^n$.
Consider a linear coding scheme in which $H_0$ is employed as the parity check matrix, and $A_1'$ is employed for recovering errors from syndromes.
That is, $\Enc(x) = x  G$ for a full-rank matrix $G \in \bin^{(n-m-2c\log n) \times n}$ satisfying $GH_0^T = 0$,
and $\Dec(y) = (y - f(A_1'(h_0, y H_0^T)) ) G^{-1}$,
where $G^{-1} \in \bin^{n \times (n-m-2c\log n)}$ is a right inverse matrix of $G$.
Then, for any $x \in \bin^m$,
\begin{align*}
&\Pr_{y \sim U_r}[\Dec( \Enc(x) + f(y)) = x] \\
  & = \Pr_{y \sim U_r}\left[
    \Enc(x)+f(y) 
    - f(A_1'(h_0, (\Enc(x)+f(y))H_0^T)) = xG 
    \right]\\
& = \Pr_{y \sim U_r}[f(A_1'(g(y,h_0))) = f(y)],
\end{align*}
where we use the property that $GG^{-1}=I$, $\Enc(x) = xG$, $GH_0^T = 0$, and $x H_0^T = h_0(x)$.
Since the probability that $g(A_0(g(y,h_0))) = g(y,h_0)$ is at least $1 - 2/n^c$,
and for any $y \in \bin^r$ satisfying $g(A_0(g(y,h_0))) = g(y,h_0)$, $\Pr_y[ f(y) \notin C_{h_0}] \geq 1 - 1/n^c$,
we have that
\begin{equation*}
\Pr_{y \sim U_r}[f(A_1'(g(y,h_0))) = f(y)] \geq 1 - \frac{3}{n^c}.
\end{equation*}
Thus, $Z$ can be corrected with error probability $O(n^{-c})$.
Since $f$ is efficiently computable, 
both $\Enc$ and $\Dec$ can be computed in time polynomial in $n$.
Hence the statement follows.
\end{proof}

It is known that if one-way functions do not exist,
then neither do distributionally one-way functions~\cite{IL89}.
Thus, Theorem~\ref{thm:distow} implies the following corollary.
\begin{corollary}
  If one-way functions do not exist,
  every samplable flat distribution over $\bin^n$ of entropy $m$ can be corrected
  by an efficient coding scheme of rate $1 - m/n - (2c\log n)/n$ and error probability $O(n^c)$ for any constant $c > 0$.
\end{corollary}
The above corollary indicates the necessity of one-way functions for proving the impossibility results of Sections~\ref{sec:pseudorandom} and~\ref{sec:membership}.


\section{Conclusions}

In this work, we study the correctability of samplable additive errors with unbounded error rate.
We have considered a relatively simple setting in which the error distribution is identical for every coding scheme and codeword.
The results imply that even when a distribution is not pseudorandom by  membership test, 
it is difficult to correct every such samplable distribution by efficient coding schemes.
Nevertheless, a positive result can be obtained if we consider much more structured errors such as errors from linear subspaces.
We present some possible future work of this study.

\setcounter{paragraph}{0}

\paragraph{Further study on the correctability.}



In this work, we have mostly discussed impossibility results.
Thus, showing non-trivial possibility results is interesting.
A possible direction is to consider more structured errors than samplable errors.
One can consider \emph{computationally} structured errors
such as errors computed by log-space machines, constant-depth circuits, or monotone circuits.
Also, one can consider other types of structures, e.g., errors are introduced in a \emph{split-state} manner.
Namely, an error vector is split into several parts, and each part is independently computed.
This model has been well-studied in the context of leakage-resilient cryptography~\cite{DP08,LL12} and non-malleable codes~\cite{DPW10,CG17,ADL14}.
BSC can be seen as an extreme of this type of channels in which each error bit is computed by the same biased-sampler.

\paragraph{Characterizing correctability.}

We have investigated the correctability of samplable additive errors
using entropy as a criterion.
There may be another better criterion for characterizing the correctability of these errors,
which might be related to efficient computability, to which samplability is directly related.
Since we have considered general distributions as error distributions,
the information-spectrum approach~\cite{VH94,Han03} may be more plausible.

\section*{Acknowledgment}
  This research was supported in part by
Japan Society for the Promotion of Science (JSPS) and
the Ministry of Education, Culture, Sports, Science and Technology (MEXT) 
Grant-in-Aid for Scientific Research Numbers 
25106509, 15H00851, 16H01705, 17H01695. 
We thank anonymous reviewers for their helpful comments.



%

\bibliographystyle{abbrv} 
\bibliography{mybib}

\end{document}